\documentclass[%
 reprint,
superscriptaddress,
nofootinbib,
 amsmath,amssymb,
 aps,
]{revtex4-2}

\usepackage[english]{babel}
\usepackage[utf8]{inputenc}
\usepackage{setspace}
\usepackage{pgfplots}
\usepgfplotslibrary{fillbetween} 
\pgfplotsset{compat=newest}
\graphicspath{ {./figures/} }
\usepackage{caption}
\usepackage{subcaption}
\usepackage{amsthm}
\usepackage{braket}
\usepackage{dsfont}
\usepackage{xcolor}
\usepackage{mathrsfs} %For \mathscr
\usepackage{tikz}
\usetikzlibrary{calc}
\usetikzlibrary{arrows}
\usepackage{tcolorbox}
\usepackage{enumitem}
\newlist{steps}{enumerate}{1}
\setlist[steps, 1]{label = Step \arabic*:}
\usepackage[colorinlistoftodos,prependcaption]{todonotes}
\newtheorem{theorem}{Theorem}[section]

\newtheorem{thm}[theorem]{Theorem}

\newtheorem{lem}[theorem]{Lemma}
\newtheorem{prop}[theorem]{Proposition}

\usepackage{mathtools}
\DeclarePairedDelimiter\ceil{\lceil}{\rceil}

\newcommand{\ketbra}[2]{\ket{#1}\!\bra{#2}}

\newcommand\norm[1]{\lVert#1\rVert}

  \newcommand{\SubItem}[1]{
    {\setlength\itemindent{15pt} \item[-] #1}}

\usepackage{graphicx}
\usepackage{dcolumn}% Align table columns on decimal point
\usepackage{bm}% bold math
\begin{document}

%\preprint{One-Way Hashing Method with Finite Resources}

%\title{One-Way Hashing Method with Finite Resources}
%\title{Deterministic entanglement distillation with finite resource}
%\title{Deterministic entanglement distillation with finite atomic qubits}
\title{Deterministic high-rate entanglement distillation with neutral atom arrays}

% Force line breaks with \\

%\author[1]{Thomas A. Hahn}
%\author[2]{Ryan White}
%\author[3]{Hannes Bernien}
%\author[1]{Rotem Arnon-Friedman}

%%%%%% Affiliations %%%%%%
%\affil[1]{The Center for Quantum Science and Technology, Department of Physics of Complex Systems, Weizmann Institute of Science, Rehovot, Israel}
%\affil[2]{
%Department of Physics, University of Chicago, Chicago, IL 60637, USA}
%\affil[3]{Pritzker School of Molecular Engineering, University of Chicago, Chicago, IL 60637, USA}

\author{Thomas A. Hahn}
 \email{thomas.hahn@weizmann.ac.il}%Lines break automatically or can be forced with \\
 \affiliation{%
 The Center for Quantum Science and Technology \\ Department of Physics of Complex Systems, Weizmann Institute of Science, Rehovot, Israel
 }
 
\author{Ryan White}
 \affiliation{Department of Physics, University of Chicago, Chicago, IL 60637, USA}
 
\author{Hannes Bernien}
 \affiliation{Pritzker School of Molecular Engineering, University of Chicago, Chicago, IL 60637, USA}
 
\author{Rotem Arnon-Friedman}
 \affiliation{%
 The Center for Quantum Science and Technology \\ Department of Physics of Complex Systems, Weizmann Institute of Science, Rehovot, Israel
 }%

\date{\today}% It is always \today, today,

\begin{abstract}
The goal of an entanglement distillation protocol is to convert large quantities of noisy entangled states into a smaller number of high-fidelity Bell pairs. The celebrated one-way hashing method is one such protocol, and it is known for being able to efficiently and deterministically distill entanglement in the asymptotic limit, i.e., when the size of the quantum system is very large. In this work, we consider setups with finite resources, e.g., a small fixed number of atoms in an atom array, and derive lower bounds on the distillation rate for the one-way hashing method. We provide analytical as well as numerical bounds on its entanglement distillation rate -- both significantly tighter than previously known bounds. We then show how the one-way hashing method can be efficiently implemented with neutral atom arrays. The combination of our theoretical results and the  experimental blueprint we provide indicate that a full coherent implementation of the one-way hashing method is within reach with state-of-the-art quantum technology.
\end{abstract}

\maketitle

%\tableofcontents

\section{\label{sec:level1}Introduction}
 Long-distance quantum networks require high fidelity quantum entanglement as a key ingredient~\cite{kimble2008quantum,simon2017towards,wehner2018quantum,cirac1997quantum,briegel1998quantum,sangouard2011quantum,pompili2021realization}. In particular, each node in the network needs to share high-quality entangled states with all other nodes. 
        This can be achieved using two steps. The first step is the initial distribution of entanglement between the faraway nodes. The form of the distribution step depends on the used technology and, fundamentally, ends with \emph{noisy} entangled states between the nodes. For example, a quantum network might consist of individually-trapped atoms within high-finesse optical cavities, connected with a network of optical fibers~\cite{kimble2008quantum, rempenetwork2012}. Remote entanglement can be generated via coherent exchange of single photons, but these operations are fundamentally imperfect. The ultimate entanglement fidelity will be limited by any number of noise sources which scramble the quantum state of the atoms and/or the exchanged photons.

        The second step is to apply an \emph{entanglement distillation protocol} (EDP)-- a concept first presented in the 90's~\cite{Bennett_1996,Bennett1996MixedstateEA} (originally termed entanglement purification). The goal of an EDP is to transform the noisy entanglement between the nodes to a high quality one. 
        Let us consider the most simple case, in which two nodes share a state of the form $\rho_{AB}^{\otimes n}$, i.e., $n$ independent and identically distributed (IID) copies of some state $\rho_{AB}$, such that 
        \begin{eqnarray}
          \operatorname{F}\left( \rho_{AB} ,\ketbra{\phi^{+}}{\phi^{+}}_{AB}\right) = 1-\epsilon \;,
        \end{eqnarray}
        where $\operatorname{F}$ is the fidelity and $ \ket{\phi^{+}} := \frac{1}{\sqrt{2}} \left[ \ket{00} + \ket{11}\right]$ is a maximally entangled state, also called a Bell state~\cite{nielsen_chuang_2010}.
        The goal of an EDP is then to transform the input state~$\rho^{\otimes n}_{AB}$ to $m\leq n$ Bell states with higher fidelity~\cite{Bennett_1996,Bennett1996MixedstateEA,khatri2020principles}.
        During the protocol the nodes can apply \emph{local quantum operations}, i.e., each node can act only on its part of the state. In addition, they can communicate classically. 
        After a successful execution of an EDP, the nodes can use their high-fidelity entanglement for the application of their  choice over the quantum network. 
        Thus, being able to apply a good EDP is mandatory for any functional quantum network. 
        
        \begin{figure}[h]\centering 
\includegraphics[scale=0.8]{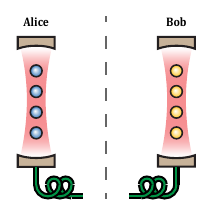}
\caption{During the protocol, each party only has access to their part of the initial quantum state.}
\label{fig:setupmodel}
\end{figure}
        Noisy Bell states are often parameterized via depolarizing noise, e.g.\
        \begin{eqnarray} \label{1Bellpair}
              \rho_{AB}:= W\ketbra{\phi^{+}}{\phi^{+}}+\frac{1-W}{4}\mathds{1}_{4} \; ,
        \end{eqnarray} 
        which models the noise as effectively being uniformly random.\footnote{Our analysis below is valid for more general states, in which the initial global state is Bell diagonal. We present the above model of many copies that are each Bell diagonal for simplicity.} 
        %For systems which do not produce qubit pairs of this form, one may apply twirling operators~\cite{Bennett1996MixedstateEA}.  
        For many independent copies of the state in Eq.~(\ref{1Bellpair}), entanglement can be distilled via the famous one-way hashing method~\cite{Bennett_1996,Bennett1996MixedstateEA}. 
        The protocol is presented in Fig.\  \ref{HashingMethod} for completeness.
        From an experimental point of view, the hashing protocol is very appealing: 
        Apart from the post-processing step, this protocol is state-independent, i.e. the applied unitaries and measurements are independent of the input state.\footnote{This is in stark contrast to most recurrence protocols~\cite{Bennett1996MixedstateEA, D_r_2007} and protocols based on fixed error correcting codes~\cite{Roque:2023qcj}.} Moreover, for any number of initial copies, the protocol is easily expressed by elementary one- and two-qubit gates.
        
        In addition, the protocol has the benefit of being very efficient at large scales; for sufficiently large $n$ IID copies of some state $\rho_{AB}$, the one-way hashing method can produce approximately 
        \begin{equation} \label{Eq: AssymptoticBellPairs}
          m \approx \left[1-\operatorname{H}(AB)_\rho\right]\cdot n =- \operatorname{H}(A|B)_\rho \cdot n 
        \end{equation}
        Bell pairs, where $\operatorname{H}(A|B)_\rho$ is the conditional von Neumann entropy~\cite{Bennett_1996,Bennett1996MixedstateEA}; this approximation is tight in the limit $n \to \infty$.
        
        In all of the presented equations, $n$ is the number of the initial weakly-entangled pairs shared by the nodes. In the case of an architecture based on atom arrays, for example, $n$ is (at most) the number of atoms in each array. With experimental feasibility in mind, considering the performance of the protocol only for large $n$ is not enough-- scalability issues of current quantum technologies prevent us from creating systems in which $n$ is very large. Thus, it is of key importance to understand how many Bell pairs can be produced via the one-way hashing method for small values of $n$.
        
        While the aim of the original papers~\cite{Bennett_1996,Bennett1996MixedstateEA} presenting the hashing method was to primarily show the asymptotic behavior of the protocol, their methods can in principle be used to bound the number of Bell states that can be produced for any value $n$. Building on the techniques from~\cite{Bennett_1996,Bennett1996MixedstateEA}, a thorough analysis on the achievable distillation rate, i.e.\ the number of produced output pairs~$m$ divided by the number of initial copies~$n$, for finite~$n$ was derived in~\cite{Zwerger_2018}. Although ~\cite{Bennett_1996,Bennett1996MixedstateEA, Zwerger_2018} all produce the same asymptotic rate, for reasonable choices of initial depolarizing noise, their analysis can only guarantee a non-trivial distillation rate once $n \sim 100$, which is significantly beyond current system sizes.

        %While the aim of~\cite{Bennett_1996,Bennett1996MixedstateEA} was to primarily show the asymptotic behavior of the one-way hashing method, their methods can in principle be used to bound the number of Bell states that can be produced for any value $n$. Building on the methods from~\cite{Bennett_1996,Bennett1996MixedstateEA}, a thorough analysis on the achievable distillation rate, i.e.\ the number of produced output pairs $m$ divided by the number of initial copies~$n$, for finite $n$ can be found in~\cite{Zwerger_2018}. Although their results can reproduce the asymptotic behavior from~\cite{Bennett_1996,Bennett1996MixedstateEA}, for reasonable choices of initial depolarizing noise, their analysis can only guarantee a non-trivial distillation rate once $n \sim 100$ (See Figure \ref{fig:Introduction}).

In this work, we give significantly improved analytical and numerical bounds on the number of Bell pairs that can be produced via the one-way hashing method in the limit of finite resources. We show that the number of  qubit pairs needed to distill a single, highly entangled, Bell pair can be significantly reduced to $n\sim 10$. This puts practical EDPs within reach of current devices. We give a blueprint for implementing the protocol on neutral atom arrays which have emerged as one of the leading platforms for quantum information processing and quantum networks~\cite{Bluvstein2024,Covey2023}.

%Moreover, our results show that the hashing method can, in fact, distill entanglement for small $n\sim 10$!  

\begin{figure}[h]\centering 
\subfloat[Lower Bound on Rate]{\label{2a}
\includegraphics[scale=0.8]{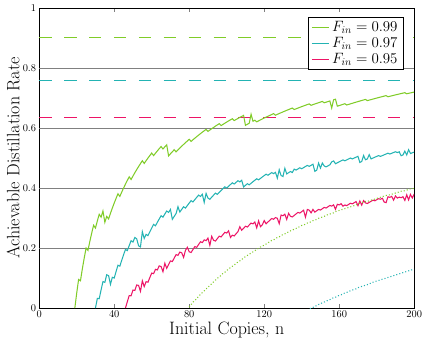}
} \\%\hfill
\hspace*{-0.22cm}
\subfloat[Numerical Simulation of Hashing Method]{\label{b}
\includegraphics[scale=0.8]{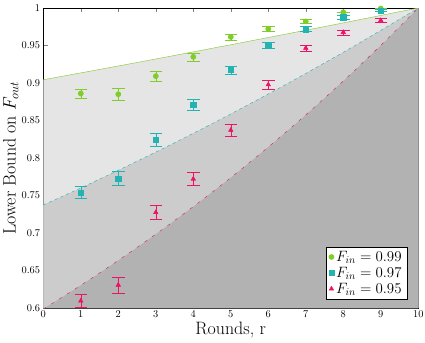}
}
\caption{In (a), the dotted lines represent previous results from~\cite{Zwerger_2018}, and solid lines represent the bounds we achieve in this work. The data points in (b) are average output fidelities over $1000$ simulations of the one-way hashing method for $n=10$ initial copies. Entanglement is distilled, once they cross their respective dashed lines.}
\label{fig:Introduction}
\end{figure} 

Fig.\ \ref{fig:Introduction} compares our results to the previous works. In Fig.\ \ref{fig:Introduction}~(a) we present the bounds on the achievable distillation rate from~\cite{Zwerger_2018} (dotted lines) and our analytical results (solid lines). One can clearly see that, for realistic input fidelities (describing the noise level of the system), the required initial copies reduces significantly from $n \sim 80-140$ to $n \sim 20-30$, and the distillation rate improves accordingly. The horizontal dashed lines represent the asymptotic value, which can be derived using Eq.\ (\ref{Eq: AssymptoticBellPairs}). The numerical data in Fig.\ \ref{fig:Introduction}~(b) indicates that entanglement can in fact be distilled for even smaller values of~$n$. More plots elucidating the strength of our work are available in Section \ref{Results}.

The results presented in this work indicate that a first proof-of-principle \emph{high-rate} entanglement distillation protocol can be achieved with state of the art technology, such as atom arrays. Atom array systems have demonstrated raw Bell state fidelities of 98\% and two-qubit CZ gate fidelities of 99.5\%~\cite{lukinhighfidelity2023}, coupling of multiple atoms to an optical network~\cite{rempenetwork2012, vuletic5atom2024}, and array sizes of thousands of atoms~\cite{endres6100atoms2024}. However, no atom array experiment has demonstrated an EDP, largely because it is difficult to achieve the necessary level of control over a large number of qubits ($n$ in our context).
%Thus, it should already be experimentally feasible to run an EDP with a small number of remotely entangled pairs, or a larger number of locally entangled pairs.
There has been an experimental demonstration of an EDP with NV center qubits~\cite{hansondistillation2017}, using two entangled pairs and the recurrence method, which does not boast the distillation rate of the one-way hashing method. A realization of the one-way hashing method (presented in Fig.\ \ref{HashingMethod}) would thus represent a significant step towards the construction of large-scale quantum networks. By showing that this distillation method is practical for fewer numbers of initial entangled pairs, our work suggests that such a realization is well within reach of modern experiments.

In the following section we explain our findings in more detail-- we present the theoretical results (with complete proofs in the Methods section) and then suggest an experimental setup for implementing the distillation protocol using atom array technology.

\section{Results}\label{Results}
\subsection{Theoretical results}\label{Theory}
In this section, we  generally consider Bell-diagonal states,\footnote{These are states that can be expressed as convex combinations of tensor products of Bell states.} which include states such as 
\begin{align} \label{Eq: ManyIIDCopies}
    \rho_{A^nB^n}= \left(W\ketbra{\phi^{+}}{\phi^{+}}+\frac{1-W}{4}\mathds{1}_{4} \right)^{\otimes n} \; .
\end{align}
For Bell-diagonal states, the noise on the global system is described by the distribution of its eigenvalues. The larger an eigenvalue is, the more probable it is that the error, which is associated to it, occurs. As was originally shown in~\cite{Bennett_1996,Bennett1996MixedstateEA}, the one-way hashing method is specifically tailored for Bell-diagonal states. Moreover, it can correct such errors precisely because the noise is described classically and the corresponding error can be expressed as multiple copies of Bell states.

If one considers the state described by Eq.\ (\ref{Eq: ManyIIDCopies}), for example, as long as the noise per entangled qubit pair is not too large, many of the eigenvalues will be so close to zero, they can effectively be ignored. In other words, one can neglect the eigenvalues associated to higher order errors, which occur at a significantly smaller rate than others.

 Informally, the goal of this work is to find tight bounds on the set of relevant errors,\footnote{This set has the property that the probability of an error, which is not included in this set, occurring is close to zero.} using only entropic quantities. Once we manage to characterize the size of this set, we use this quantity to generate lower bounds on the number of Bell pairs that can be distilled using the one-way hashing method.
 To do this, one uses the fact that, after every round, roughly half of the errors in this set will not be compatible with the observed measurement outcomes, i.e.\ the protocol reduces the number of possible errors by approximately a factor of two. We then require that the protocol only stops when one error remains in this set. Once this happens, both parties 

\begin{tcolorbox}[size=small,title=One-Way Hashing Method: (Alice and Bob initially share n qubit pairs)] 
Round $k+1$: Alice and Bob share $n-k$ qubit pairs
\begin{steps} [leftmargin=1.4cm] 
  \item Alice and Bob generate a uniformly random bitstring $S = s_{1}\cdots s_{n-k}$ where each $s_{j}$ represents two bits. Let $s_{j^{\star}}$ represent the first non-zero 2-bit string.%\footnote{If no such 2-bit string exists, one can either generate a new bitstring $S$ or randomly select a value for $j^{\star}$, and jump directly to Step 4.}
  \item For all $j \in \{ 1,\dots, n-k \}$, if 
 \SubItem{$s_{j} = 10$, then both Alice and Bob apply a $\pi/2$-rotation around the y-axis on their half of the $j$'th qubit pair.}
   \SubItem{$s_{j} = 11$, then Alice and Bob, respectively, apply a $3\pi/2$- and $\pi/2$-rotation around the x-axis on their half of the $j$'th qubit pair.}
  \SubItem{$s_{j} = 00$ or $s_{j} = 01$, no actions are required at Step 2.}
  \item For all $j \neq j^\star$ s.t. $s_{j} \neq 00$, both Alice and Bob apply a CNOT gate on qubit pairs $j$ and $j^\star$, where pair $j^\star$ contains the target qubits.
  \item Alice and Bob measure qubit pair $j^\star$ in the computational basis and discard said pair. Alice broadcasts her measurement outcomes to Bob.
\end{steps}
Post-processing: After all rounds are concluded, Alice and Bob share a Bell-diagonal state. Based on the initial (pre-protocol) bipartite state and all past joint measurement outcomes, Bob applies single-qubit Pauli gates such that the largest eigenvalue of the post-processed state corresponds to multiple copies of the Bell state $\ketbra{\phi^{+}}{\phi^{+}}_{AB}$. 
\end{tcolorbox}
\noindent\begin{minipage}{\columnwidth}
\captionof{figure}{The one-way hashing method.}\label{HashingMethod}
\end{minipage}
have successfully detected the error, and they can then go on to correcting it, using only Pauli gates (See the post-processing step presented in Fig.\ \ref{HashingMethod}).

We discuss this now more formally. Let $\mathcal{H}_{A^nB^n} $ represent the Hilbert space of the two parties, i.e.\ each party holds $n$ qubits, and let $\rho_{A^nB^n} \in S_{=} \left(\mathcal{H}_{A^nB^n}  \right)$ be a (normalized) bipartite density matrix, that is diagonal in the Bell basis. Furthermore, let us denote by $P_{X^n}$ the distribution of the eigenvalues of $\rho_{A^{n}B^{n}}$. The Hartley entropy of $P_{X^n}$ is defined as~\cite{Tomamichel2015QuantumIP,Tomamichel2015QuantumIParticle}
 \begin{align}
     H_{0}\left(X^n\right)_P := \log_{2} |\{x: P_{X^n}(x) >0\}| \; .
 \end{align}
Given an initial state $\rho_{A^nB^n}$, let $m^\epsilon$ be the maximal number Bell pairs, up to a negligible (global) error $\epsilon$, that can be created via the one-way hashing method. The rank of the distribution $P_{X^n}$  gives us information as to how many errors need to be corrected such that the final state is pure. The lower bound in Eq.\ (\ref{nontightbound}) then follows from 
 \begin{figure*}[t]
 \centering
\subfloat[Lower Bound on Rate]{\label{2a}
\includegraphics[scale=.73]{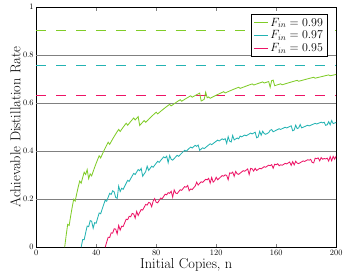}
}%\hfill
\subfloat[Upper Bound on Single-Pair Distillation]{\label{2b}
\includegraphics[scale=.73]{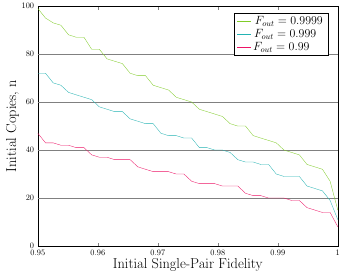}
}%\par 
\subfloat[Comparison: Rate]{\label{2c}
\includegraphics[scale=.73]{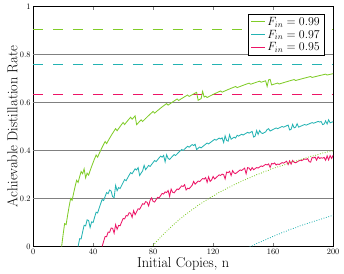}
}%\hfill
\subfloat[Comparison: Single-Pair Distillation]{\label{2d}
\includegraphics[scale=.73]{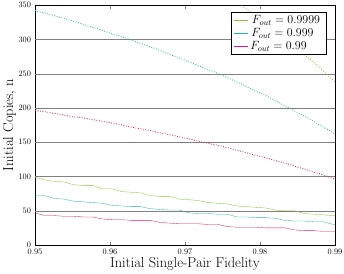}
}
\caption{In all subfigures, the initial state is $n$ IID copies of Werner states. In (a), we compare the achievable distillation rates between Werner states of different fidelities. The output fidelity is required to be at least $F_{out} = 0.99$. The horizontal lines represent the asymptotic distillation rate. In (b), we compare upper bounds on the initial copies needed to achieve a single Bell pair for different fidelities of the output state. We compare our results to the best previously known bounds in (c) and (d).}
\label{fig:analytics}
\end{figure*} 

\noindent the fact that, during each round of the one way-hashing method, the number of errors essentially reduces by a factor of two. (A proof is given in the Methods section).

\begin{equation} \label{nontightbound}
    m^\epsilon \geq  n-\lceil H_{0}(X^{n})_{P}-2\log_{2}\left(\epsilon\right) \rceil \; .
\end{equation}
To see why Eq.\ (\ref{nontightbound}) is in general \emph{not tight}, note that $n$ IID copies of states in the form of Eq.\ (\ref{1Bellpair}) will have maximal rank for $W<1$ and Eq.\ (\ref{nontightbound}) then yields negative lower bounds. It is not required for the hashing method to correct all errors, though. Rather, it is sufficient for the protocol to just correct the most prevalent errors. In this case, the Hartley entropy cannot be used to quantify the number of errors that need to be corrected and it has to be replaced with the so-called \textit{smooth} Hartley entropy, $H_{0}^{\epsilon}$; a formal definition is given in Methods. In a similar vein to Eq.\ (\ref{nontightbound}), $m^\epsilon$ can be lower bounded using the smooth Hartley entropy.
\begin{thm}  \label{outputtightbound}
Let $\rho_{A^nB^n} \in S_{=} \left(\mathcal{H}_{A^nB^n}  \right)$ be a normalized Bell-diagonal density matrix, whose eigenvalues are described by a probability distribution $P_{X^n}$. For all $\epsilon_1,\epsilon_2 >0$ that satisfy $\epsilon_{1}+\epsilon_{2} \leq \epsilon$, 
\begin{equation}
    m^\epsilon \geq  n-\lceil H_{0}^{\epsilon_{1}}(X^{n})_{P}-2\log_{2}\left(\epsilon_{2}\right) \rceil \; .
\end{equation}
\end{thm}

For $n$ IID copies of some state $\rho_{AB}$ as an initial state, the distillation rate, $R^\epsilon$, of the one-way hashing method is defined as the ratio between the number of Bell pairs the protocol outputs and $n$, i.e.
\begin{equation}
    R^\epsilon := \frac{m^\epsilon}{n} \; .
\end{equation}
The rate should be thought of as a ``measure" for how much entanglement can be extracted from each of the $n$ copies of $\rho_{AB}$, and it can be lower bounded using Theorem \ref{ratetightbound}.
\pagebreak
\begin{thm}[Rate lower-bound]  \label{ratetightbound}
Let ${\rho_{A^nB^n} \in S_{=} \left(\mathcal{H}_{A^nB^n}  \right)}$ be a normalized Bell-diagonal density matrix, whose eigenvalues are described by a probability distribution $P_{X^n}$. For all $\epsilon_1,\epsilon_2 >0$ that satisfy $\epsilon_{1}+\epsilon_{2} \leq \epsilon$, the rate of the one-way hashing method can be lower bounded by
\begin{align}
    R^\epsilon \geq \frac{n-\ceil*{ H_{0}^{\epsilon_{1}}(X^{n})_{P}-2\log_{2}\left(\epsilon_{2}\right) }}{n} \; .
\end{align}
\end{thm}
To derive explicit bounds on the distillation rate via Theorem \ref{ratetightbound}, one first needs to calculate $H_{0}^{\epsilon_{1}}(X^{n})_{P}$. We do this similarly to~\cite{RenWol05}, which uses a slightly different definition for the smooth Hartley entropy.\footnote{They use the generalized trace distance.}
%, as opposed to the generalized fidelity
In particular, we show that the smooth Hartley entropy is equivalent to the following discrete optimization problem.

%\pagebreak
\begin{lem} \label{OptProbHartley}
Let $\mathcal{X}^n$ be a finite set and let $P_{X^n}(x)$ be a normalized probability distribution on $\mathcal{X}^n$. Then $H_{0}^{\epsilon} \left(X^n\right)_P$ is equal to
\begin{equation} \label{optprobHart}
 \begin{aligned}
\min_{k} \quad & \log_{2} \left(k\right) \\
\textrm{s.t.} \quad & \sum_{x \in \mathcal{I}_{k}} P_{X^n}\left(x\right) \geq 1 - \epsilon^2 \; ,
\end{aligned}
\end{equation}
where $P_{X^n}\left(x_{1}\right), \dots,  P_{X}\left(x_{k}\right)$ are the $k$ largest weights of the distribution and $\mathcal{I}_{k} =\{ x_{1}, \dots,  x_{k}\}$.
\end{lem}
Fig.\ \ref{2a} compares bounds on the distillation rate given by Theorem \ref{ratetightbound} for $n$ IID copies of states in the form of Eq.\ (\ref{1Bellpair}). Theorem \ref{ratetightbound} also implicitly gives an upper bound on the number of copies needed to produce a single, high-fidelity Bell pair.\footnote{The upper bound is the minimal $n$ needed s.t. the bound on the distillation rate is positive.} These upper bounds are 
\begin{figure*}\centering
\subfloat[Hashing Method for $n=6$]{\label{a}
\includegraphics[scale=.8]{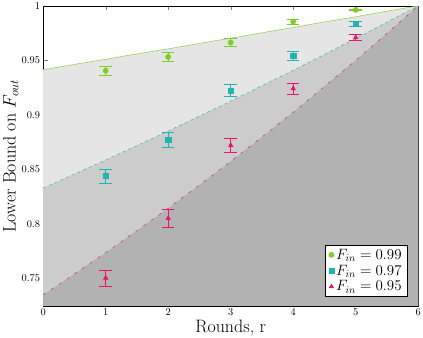}
}%\hfill
\subfloat[Hashing Method for $n=8$]{\label{b}
\includegraphics[scale=.8]{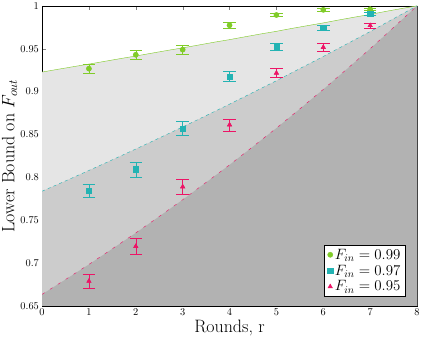}
}%\par 
\subfloat[Hashing Method for $n=10$]{\label{c}
\includegraphics[scale=.8]{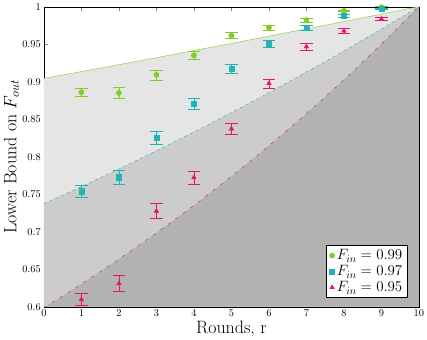}
}
\caption{We numerically simulate the hashing method for $n \in \{6,8,10\}$ IID initial copies of Werner states. The fidelity of each initial Werner state copy is denoted by $F_{in}$. Each data point represents a lower bound on the average output fidelity over $r=1000$ simulations. The output state has a higher fidelity than $n-r$ copies of the initial Werner state, if the data point lies above their respected line.}
\label{fig:numerics}
\end{figure*}

\noindent displayed in Fig.\ \ref{2b}. Figs.\ \ref{2c} and \ref{2d} compare our results to  known bounds for the hashing method from the previous work~\cite{Zwerger_2018}. They show that bounds via the smooth Hartley entropy are far more precise, even when $n \sim 100$. We remark that the rates are comparable to those achieved by protocols based on
error correcting codes~\cite{Roque:2023qcj}, which achieve high rates, but are highly state dependent.

Similar to~\cite{Zwerger_2018}, our results also converge to the asymptotic rate from~\cite{Bennett_1996,Bennett1996MixedstateEA}. This follows from both the classical~\cite{holenstein2006randomness} and fully quantum asymptotic equipartition property (AEP)~\cite{5319753} . Applying either AEP to the bound from Theorem \ref{ratetightbound} directly yields Lemma \ref{rateasymptound}.
\begin{lem}  \label{rateasymptound}
Let $ \rho_{AB}^{\otimes n} \in S_{=} \left(\mathcal{H}_{A^nB^n}  \right)$ be a normalized Bell-diagonal density matrix, whose eigenvalues are described by a probability distribution $P_{X^n}$. Then
\begin{equation}
    \lim_{\epsilon_1, \epsilon_2 \to 0}  \lim_{n \to \infty}  \frac{n-\ceil*{H_{0}^{\epsilon_{1}}(X^{n})_{P}-2\log_{2}\left(\epsilon_{2}\right) }}{n} = -H(A|B)_\rho \; .
\end{equation}
\end{lem}

While it can thus be said that Theorem \ref{ratetightbound} is asymptotically tight, this does not mean that this lower bound is always optimal for fixed number of initial copies, $n$. Fig.\ \ref{fig:numerics} shows the results from numerical simulations of the one-way hashing method for $n \in \{ 6,8,10\}$, where each data point represents the average over $1000$ simulation runs.\footnote{To improve the runtime, the data is generated by simulating the protocol on a state close to the original state in purified distance. The numerical results are then used to bound the output fidelity of the original state via the triangle property of the purified distance.} These results indicate that, for small $n$, the actual achievable distillation rate is in fact much better than the analytical lower bound. Entanglement can be said to be distilled once the data points lie above their corresponding dotted lines. This happens after relatively few steps, thus strongly suggesting that the one-way hashing method is indeed experimentally feasible.

\subsection{Experiment Proposal}\label{Experiment}
Neutral atom processors encode quantum information in the long-lived electric and nuclear states of single atoms, typically an alkali or alkaline earth element~\cite{kangkuenreview2021}. Arrays of thousands of atomic qubits have been demonstrated~\cite{endres6100atoms2024}, making neutral atoms a strong option for large-scale quantum computing tasks. Furthermore, control of the qubit --- including state preparation, qubit manipulation, and measurement --- can be performed using lasers tuned to atomic transitions, resulting in fast, high-fidelity operations ($>$ 99.98\% fidelity 1-qubit rotations~\cite{lukinraman2022}). Despite the impressive developments of this platform, only tens of entangled pairs have been produced simultaneously in an experiment~\cite{lukinhighfidelity2023}, limited by factors such as a minimum spacing between entangled pairs, and trade-offs between laser beam size and irradiance. Thus, near-term devices will make the most use of distillation protocols which demonstrate high yield even for limited numbers of initial entangled pairs.

While it is easiest to entangle neutral atoms locally (i.e. at micrometer-scale distances within a single setup), the optical nature of atomic transitions presents the opportunity to build a quantum network consisting of separate many-atom nodes --- each possessing high-fidelity local operations --- interconnected with fiber optical links~\cite{kimble2008quantum, rempenetwork2012}. This is typically achieved by coupling atoms to an optical cavity, enhancing the light-matter interaction and enabling efficient collection of photons into fibers. Since photon emission is conditional on the state of an atom, an atom prepared in a superposition state can become entangled with its emitted photon. This photon is then sent through the network, and absorbed by an atom at another node, resulting in long-distance atom-atom entanglement. Scaling to many entangled atoms can be achieved with additional techniques, such as selectively coupling multiple atoms to one cavity~\cite{vuletic5atom2024} or coupling multiple atoms to multiple cavities~\cite{berniennanophotonic2024}. A comparable experiment, performed using NV center qubits in diamonds, used similar quantum networking techniques to create two long-distance entangled qubit pairs, and distill them into a single higher-fidelity entangled pair~\cite{hansondistillation2017}. However, this demonstration utilized the recurrence method, which is known to not scale well with the number of input entangled pairs. Producing larger entangled resources will inevitably require stronger distillation procedures, such as the one-way hashing method, and a large register of high-precision qubits. Yet present-day quantum network nodes are still limited in their scale and operation fidelity, so experiments which require sharing high-fidelity entangled resources will necessitate an EDP that has reasonable yield with small numbers of input entangled pairs.

For the local 2-qubit gates required by the hashing method, neutral atom processors can take advantage of entangling gates enabled by van der Waals interactions between high-energy `Rydberg' states~\cite{lukinjaksch2000,browaeysrydberg2010}. These interactions have been extensively used to perform controlled-phase gates between qubits, and recent results using optimal control techniques have achieved $\geq$ 99.5\% gate fidelity~\cite{lukinhighfidelity2023}. Of course, such a gate can also be used to generate many entangled pairs within a single quantum processor; applying a local distillation procedure would then provide a resource of high-quality entangled states for subsequent experimentation. 

The `natural' two-qubit gate for Rydberg-based systems is CZ, but the hashing method is written in terms of CNOT. One could simply decompose the CNOTs into CZs and Hadamards, but to reduce experimental requirements we can instead reformulate the algorithm in terms of CZ. This introduces minor changes to the single-qubit gates, but makes no changes to the theoretical performance (see appendix). It is worth noting that Rydberg gates do not provide all-to-all connectivity as required by the hashing method, so additional control techniques are required. Some solutions include dynamically rearranging the atoms to change the qubit connectivity~\cite{lukintransport2022}, using a second atomic species as an auxiliary qubit to mediate longer-range gates~\cite{bernienrydberg2024}, or simply compiling additional CZ gates to bridge the gap. An example of a CZ-based round of the one-way hashing method is given in Fig.\ \ref{fig:experiment}.

Another experimental hurdle is the targeted addressing of gates. Neutral atom qubit control is typically applied `globally', i.e. to every atom simultaneously. Of course, most circuits require specific gates to be applied to specific qubits. This introduces additional complexity to the optical control systems, but solutions have been demonstrated. Similarly to the qubit connectivity issue, one solution is to dynamically rearrange atoms; moving the atoms in or out of an `interaction region' allows gates to be applied to some atoms without affecting others \cite{lukintransport2022}. Alternatively, tightly-focused lasers could be steered onto specific atoms, implementing a gate only on those sites~\cite{saffmanalgorithms2022}. It may be the case that, for example, only site-selective Z and CZ operations are available, in which case site-selective X and Y operations can be decomposed into selective Zs and global Xs and Ys~\cite{bakerdecomposing2023}.
\begin{figure}[h]\centering
\includegraphics[scale=1]{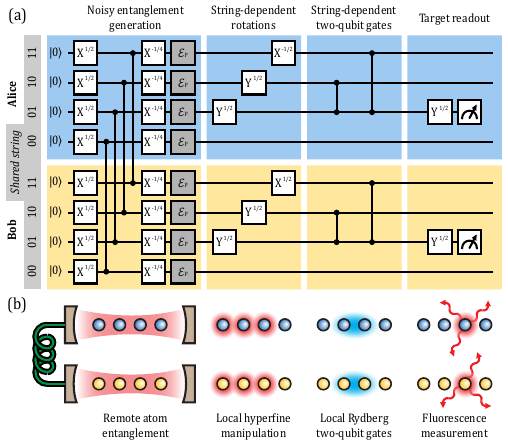}
\caption{Example of one round of the hashing method. (a) The bitstring shared by Alice and Bob determines what local gates to apply, and which entangled pair to measure. The hashing method can be modified for different initial entangled states (e.g. $\ket{\phi^+}$) and different sets of gates (e.g. using CZ instead of CNOT) without affecting any of the discussed bounds (see Section \ref{Subsection: ModifiedProtocol}). (b) Modern quantum computing systems, such as those based on neutral atoms, should be able to perform a modestly-sized distillation. Non-local entanglement can be generated across an optical quantum network, and laser control at each node enables local operations and measurement.}
\label{fig:experiment}
\end{figure} 

\section{Discussion} \label{Discussion}
As can be seen in Fig.\ \ref{fig:analytics}, Theorem \ref{ratetightbound} represents a genuine improvement on the previously known lower bounds from~\cite{Zwerger_2018}. This is due to the fact that we do not use concentration inequalities to bound the size of the typical set, which contains the errors that the one-way hashing method should be able to correct. Rather, we calculate the smooth Hartley entropy, which directly quantifies the minimal number of errors that need to be corrected for the one-way hashing method to produce a maximally entangled state~\cite{Renes_2012,tomamichel2013framework,gallager1968information}. Also, unlike concentration inequalities, $H_{0}^{\epsilon_{1}}(X^{n})$ does not require the initial state to be an IID tensor product, which is why our bounds hold for all Bell-diagonal states. 

The discrepancy between our analytical and numerical bounds is due to the proof technique that is applied. Similar to~\cite{Bennett_1996,Zwerger_2018}, we implicitly assume the worst case scenario, which is that the measurement outcomes have to provide a unique syndrome for each error on the initial state that we want to detect and subsequently correct. That is, if there exist at least two types of errors within the set of `relevant errors' that would produce the same observed measurement outcomes, then we implicitly lower bound the fidelity between $m$ Bell states and the output state of the protocol by the value $0$. This, however, does not take into account that these errors may happen with different probabilities. Nor does it account for the fact that the final state lies on a smaller Hilbert space than the original state, and therefore different initial errors may be mapped to the same error on the final state. When this occurs, the number of errors that actually need to be corrected decreases. The numerical simulations, however, do account for both of these effects, which is why they produce tighter bounds on the fidelity of the output state. 

One of the more experimentally demanding properties of the one-way hashing method is that the actions of each round depend on a random string $S$. As such, the experimental set-up must be capable of applying all of the potential gates. This issue can be mitigated in the following way. For $n$ IID copies of some state $\rho_{AB}$ of the form given by Eq.~(\ref{1Bellpair}), some choices of $S$ may be better at detecting and correcting errors than others. In Section~\ref{ECCConnection} we provide choices of $S$
that act as $[[4,2,2]]$ and $[[5,1,3]]$ stabilizer codes, for $n=4$ and $n=5$, respectively.

These results show that more efficient distillation protocols, such as the the one-way hashing method, may be implementable on modern quantum experimental platforms. Specifically, by showing that the one-way hashing method efficiently distills entanglement even for small numbers of initial entangled pairs, we suggest that near-term quantum networks have a feasible route towards distributing high-fidelity entanglement. As an example, we note that neutral atom arrays possess all of the ingredients required for a demonstration of large-scale distillation: large numbers of long-lived qubits, high fidelity one- and two-qubit operations, and the ability to generate remote entanglement by exchanging photons over a fiber optic network. By taking advantage of the one-way hashing method, these experiments could produce high-fidelity long-distance entanglement between atoms, which can subsequently be used for secure quantum communication, distributed quantum processing, enhanced metrology, and countless other applications.

\section*{Acknowledgments}
TH and and RAF acknowledge support from the Marshall and Arlene Bennett Family Research Program, the Minerva foundation with funding from the Federal German Ministry for Education and Research and the Israel Science Foundation (ISF), and the Directorate for Defense Research and Development (DDR\&D), grant No. 3426/21.3. RW is supported by a National Science Foundation Graduate Research Fellowship (Grant No. 2140001). HB and RW gratefully acknowledge funding from the NSF QLCI for Hybrid Quantum Architectures and Networks (NSF award 2016136), the NSF Quantum Interconnects Challenge for Transformational Advances in Quantum Systems (NSF award 2138068), the NSF Career program (NSF award 2238860).

\section*{Code Availability}
The MATLAB code used to generate the data can be found at:
\url{https://github.com/Thomas0501/One-Way-Hashing-Method}

\onecolumngrid
\section{Methods} \label{Methods}
\subsection{Definitions}
Over any set $\mathcal{X}$, the generalized fidelity and purified distance between two classical (sub-normalized) distributions $P_X,Q_X \in S_{\leq} \left(\mathcal{X}\right)$, are defined as
\begin{align}
    \operatorname{F}  (P_X,Q_X) &:= \left(\sum_{x \in \mathcal{X}} \sqrt{P_X(x) Q_X(x)} +\sqrt{\left(1-\operatorname{Tr}[P_X]\right)\left(1-\operatorname{Tr}[Q_X]\right)}\right)^2 \\
    \operatorname{P} (P_X,Q_X) &:=\sqrt{1-\operatorname{F}  (P_X,Q_X)} \; ,
\end{align}
respectively~\cite[Chapter 3]{Tomamichel2015QuantumIP}. These definitions can of course be generalized to quantum states. For sub-normalized density matrices $\rho_A, \ \sigma_A \in S_{\leq} \left(\mathcal{H}_A\right)$ on a Hilbert space $\mathcal{H}_A$, the generalized fidelity and purified distance are given by~\cite[Definitions 3.7 and 3.8]{Tomamichel2015QuantumIP}
\begin{align}
    \operatorname{F} (\rho_A,\sigma_A) &:=\left(\norm{\sqrt{\rho_A}\sqrt{\sigma_A}}_1 +\sqrt{\left(1-\operatorname{Tr}[\rho_A]\right)\left(1-\operatorname{Tr}[\sigma_A]\right)}\right)^2 \\
    \operatorname{P} (\rho_A,\sigma_A) &:=\sqrt{1-\operatorname{F}  (\rho_A,\sigma_A)}
    \; .
\end{align}
The entropies, which we will make use of for the proofs, are the Hartley and $\epsilon$-smooth Hartley entropies. They are defined for a sub-normalized probability distribution $Q_X$ over a set $\mathcal{X}$ as 
\begin{eqnarray}
    H_{0}\left(X\right)_P &:=& \log |\{x: P_X(x) >0\}| \\
    H_{0}^{\epsilon} \left(X\right)_P &:=& \inf_{Q_X \in  \mathscr{B}^{\epsilon} \left(P_X\right)} H_{0}\left(X\right)_Q \; ,
\end{eqnarray}
respectively, where 
\begin{eqnarray}
    \mathscr{B}^{\epsilon} \left(P_X\right) := \{Q_X \in S_{\leq} \left(\mathcal{X}\right): \operatorname{P} (P_X,Q_X) \leq \epsilon \} \; .
\end{eqnarray}
For a sub-normalized quantum state, $\rho_{A} \in S_{\leq} \left(\mathcal{H}_A\right)$,
\begin{eqnarray}
    \mathscr{B}^{\epsilon} \left(A\right)_\rho := \{\sigma \in S_{\leq} \left(\mathcal{H}_A\right): \operatorname{P} (\rho_{A},\sigma_{A}) \leq \epsilon \} \; .
\end{eqnarray}

\subsection{Technical Details}
For any Bell-diagonal state $\rho_{A^nB^n}$,  $r$ rounds of the hashing method take as input a quantum state $\rho_{A^nB^n}$ and a uniformly random bitstring, $s_{[r]}$, which dictates how the protocol acts during these rounds. When compared to the protocol description in Figure~\ref{HashingMethod}, $s_{[r]}$ should be viewed as denoting all of the bitstrings that were used for $r$ rounds of the protocol.
Moreover, it outputs $n-r$ qubit pairs, as well as measurement outputs, which Alice and Bob use to correct the occurred error (Since $s_{[r]}$ is public classical information, we include it in the protocol's output). It is thus a mapping of the form
\begin{eqnarray}
    \mathcal{L}_{hash}^{r}: \mathcal{H}_{A^{n}B^n}\otimes \mathcal{H}_{S_{[r]}} &\to&  \mathcal{H}_{A^{n-r}B^{n-r}} \otimes  \mathcal{H}_{Y} \otimes  \mathcal{H}_{S_{[r]}} \\
    \rho_{A^nB^n}  \otimes \rho_{S_{[r]}}&\mapsto& \sum_{y,s} \operatorname{Pr}(Y=y \land S_{[r]}=s_{[r]})\rho^{\prime}_{|Y=y,S_{[r]}=s_{[r]}} \otimes \ketbra{y}{y}\otimes \ketbra{s_{[r]}}{s_{[r]}} \; ,
\end{eqnarray}
where $s_{[r]}$ describes the random bitstring that was used for the $r$ rounds, $y$ represents the measurement outcomes of Alice and Bob, $\rho_{S_{[r]}}$ is a fully mixed state, and $\rho^{\prime}_{|Y=y,S_{[r]}=s_{[r]}}$ is the remaining quantum state that Alice and Bob share at the end of the protocol. Moreover, for all $s_{[r]}$ and $y$, Alice and Bob will ensure that the largest eigenvalue of $\rho^{\prime}_{|Y=y,S_{[r]}=s_{[r]}}$ corresponds to the Bell state ${ \ket{\Phi_{A^{n-r} B^{n-r}}} :=\ket{\phi^+}^{\otimes n-r}}$ in the post-processing step of the hashing protocol.

Ideally, Alice and Bob want the protocol output $\rho^{\prime}_{A^{n-r}B^{n-r}}$ (i.e. the post-protocol state after tracing out the classical registers) to be very close to the state $\ketbra{\Phi_{A^{n-r} B^{n-r}}}{\Phi_{A^{n-r} B^{n-r}}}$.
Proposition \ref{BenDistillationError} rephrases the results from~\cite{Bennett1996MixedstateEA} in terms of the (generalized) fidelity, using the Hartley entropy. 
\begin{prop} \label{BenDistillationError}
Let $\rho_{A^nB^n} \in S_{=} \left(\mathcal{H}_{A^nB^n}  \right)$ be a normalized Bell-diagonal density matrix. Then
\begin{eqnarray}
    \operatorname{F} (\rho^{\prime}_{A^{n-r}B^{n-r}},\ketbra{\Phi_{A^{n-r} B^{n-r}}}{\Phi_{A^{n-r} B^{n-r}}} ) \geq 1 - 2^{(H_0 (A^nB^n)_{\rho} -r)}
\end{eqnarray}
holds, where $\rho^{\prime}_{A^{n-r}B^{n-r}YS_{[r]}}:= \mathcal{L}_{hash}^{r}\left(\rho_{A^nB^n} \otimes \rho_{S_{[r]}} \right)$ and $\rho_{S_{[r]}}$ is the fully mixed state.
\end{prop}
\begin{proof}
    It is shown in~\cite{Bennett1996MixedstateEA} that, averaged over $Y$ and $S_{[r]}$, the output quantum state $\rho^{\prime}_{|Y=y,S_{[r]}=s_{[r]}}$ has rank $1$ with probability at least $1-2^{(H_0 (A^nB^n)_{\rho} -r)}$. For these $y$ and $s_{[r]}$, the distillation protocol succeeded, and $\rho^{\prime}_{|Y=y,S_{[r]}=s_{[r]}}$ is pure, and equal to $\ketbra{\Phi_{A^{n-r} B^{n-r}}}{\Phi_{A^{n-r} B^{n-r}}}$. We say that $(y,s) \in \mathcal{S}_{pass}$ if $\rho^{\prime}_{|Y=y,S_{[r]}=s_{[r]}}$ has rank one.

    By construction, each step of the one-way hashing method maps Bell-diagonal states to Bell-diagonal states. In particular this means that the states $\rho^{\prime}_{A^{n-r}B^{n-r}}$, $\rho^{\prime}_{|Y=y,S_{[r]}}$, and $\ketbra{\Phi_{A^{n-r} B^{n-r}}}{\Phi_{A^{n-r} B^{n-r}}}  $ commute. The relevant fidelity expression therefore simplifies to 
    \begin{align}
        \operatorname{F} (\rho^{\prime}_{A^{n-r}B^{n-r}},\ketbra{\Phi}{\Phi} ) =  \operatorname{Tr} \left[ \sqrt{\sum_{y,s} \operatorname{Pr}(Y=y \land S_{[r]}=s_{[r]}) \rho^{\prime}_{|Y=y,S_{[r]}=s_{[r]}}\ketbra{\Phi}{\Phi  }}\right]^2 \; ,
    \end{align}
where we abbreviate $\ket{\Phi_{A^{n-r} B^{n-r}}}$ with $\ket{\Phi}$. It then follows that
    \begin{align}
        &\operatorname{Tr} \left[ \sqrt{\sum_{y,s} \operatorname{Pr}(Y=y \land S_{[r]}=s_{[r]})\rho_{|Y=y,S_{[r]}=s_{[r]}}\ketbra{\Phi}{\Phi}  }\right]^2 \\
        &\geq  \operatorname{Tr} \left[ \sqrt{\sum_{(y,s) \in \mathcal{S}_{pass}} \operatorname{Pr}(Y=y \land S_{[r]}=s_{[r]})\rho_{|Y=y,S_{[r]}=s_{[r]}}\ketbra{\Phi}{\Phi}  }\right]^2 \\
        &= \left[\sqrt{\sum_{(y,s) \in \mathcal{S}_{pass}} \operatorname{Pr}(Y=y \land S_{[r]}=s_{[r]})}\right]^2 \\
        &\geq  1-2^{(H_0 (A^nB^n)_{\rho} -r)}  \; .
    \end{align}
The first inequality holds because the square root is an operator monotone (see e.g.~\cite[Proposition V.1.8]{BhatiaMatrixAnalysis}) and we are removing the positive semi-definite term
\begin{eqnarray}
    \sum_{(y,s) \notin \mathcal{S}_{pass}} \operatorname{Pr}(Y=y \land S_{[r]}=s_{[r]})\rho_{|Y=y,S_{[r]}=s_{[r]}}\ketbra{\Phi}{\Phi} \; ,
\end{eqnarray}
which is simply proportional to $\ketbra{\Phi}{\Phi}$.
The following equality results from the fact that $\rho_{|Y=y,S_{[r]}=s_{[r]}}=\ketbra{\Phi}{\Phi}$ if $(y,s) \in \mathcal{S}_{pass}$. The last inequality just uses that $(y,s) \in \mathcal{S}_{pass}$ with probability at least $1-2^{(H_0 (A^nB^n)_{\rho} -r)}$. 
\end{proof}
\begin{prop}[Non-tight Bounds]  \label{Method: PropNon-tight Bound}
Let $\rho_{A^nB^n} \in S_{=} \left(\mathcal{H}_{A^nB^n}  \right)$ be a normalized Bell-diagonal density matrix, whose eigenvalues are described by a probability distribution $P_{X^n}$. For all $ \epsilon > 0$, 
\begin{equation} \label{Method: Eqnontightbound}
    m^\epsilon \geq  n-\ceil*{H_{0}(X^{n})_{P}-2\log_{2}\left(\epsilon\right)}\; ,
\end{equation}
and the rate of the one-way hashing method can be lower bounded by
\begin{equation}
    R^\epsilon \geq \frac{n-\ceil*{ H_{0}(X^{n})_{P}-2\log_{2}\left(\epsilon\right) }}{n} \; .
\end{equation}
\end{prop}
\begin{proof}
    Let us assume that 
\begin{equation}
        n \geq \lceil H_{0}(X^{n})_{P}-2\log_{2}\left(\epsilon\right) \rceil \; ,
\end{equation}
    as the bound otherwise trivially holds.  Applying $r:= \lceil H_{0}(X^{n})_{P}-2\log_{2}\left(\epsilon\right) \rceil$ rounds of the one-way hashing method will output the state $\rho^{\prime}_{A^{n-r}B^{n-r}YS_{[r]}}:= \mathcal{L}_{hash}^{r}\left(\rho_{A^nB^n} \otimes \rho_{S_{[r]}} \right)$, where $\rho_{S_{[r]}}$ is the fully mixed state. By Proposition \ref{BenDistillationError}, the purified distance between $\rho^{\prime}_{A^{n-r}B^{n-r}}$ and the maximally entangled state is bounded by 
    \begin{align}
        \operatorname{P} (\rho^{\prime}_{A^{n-r}B^{n-r}},\ketbra{\Phi}{\Phi} ) &:= \sqrt{1-\operatorname{F} (\rho^{\prime}_{A^{n-r}B^{n-r}},\ketbra{\Phi}{\Phi} ) }\\
        &\leq \sqrt{2^{(H_0 (A^nB^n)_{\rho} -r)}} \\
        &\leq \epsilon \; .
    \end{align}
The output is thus within $\epsilon$ distance of the desired state, and both parties are left with $n-r$ qubit pairs. The optimal number of Bell pairs, $m^\epsilon$, which can be produced via the one-way hashing method is thus bounded by
    \begin{equation}
    m^\epsilon \geq  n-\lceil H_{0}(X^{n})_{P}-2\log_{2}\left(\epsilon\right) \rceil \; ,
\end{equation}
and the corresponding bound on the achievable rate is attained by dividing this by $n$.
\end{proof}
For any initial Bell-diagonal state, $\rho_{A^{n}B^{n}} \in S_{=} \left(\mathcal{H}_{A^nB^n}  \right)$, its eigenvalues can be described by a probability distribution $P_{X^n}$. We are now interested in the Bell-diagonal state $\sigma_{A^{n}B^{n}}$ that is $\epsilon$-close to our initial state and has minimal rank, i.e. $\sigma_{A^{n}B^{n}} \in \mathscr{B}^{\epsilon} \left(A^n B^n\right)_\rho$ and $H_0 (A^nB^n)_{\sigma} = H_{0}^{\epsilon} (X^n)_{P}$. %Note that, by definition, such a state must always exist. 
Lemma \ref{NormalizedOptimizer} states that one can w.l.o.g.\ assume that $\sigma_{A^{n}B^{n}}$ is normalized.

\begin{lem}\label{NormalizedOptimizer}
Let $\rho_{A^{n}B^{n}} \in S_{=} \left(\mathcal{H}_{A^nB^n}  \right)$ be a normalized Bell-diagonal quantum state, and let $\sigma_{A^{n}B^{n}} \in S_{\leq} \left(\mathcal{H}_{A^nB^n}  \right)$ be a sub-normalized Bell-diagonal state such that $\sigma_{A^{n}B^{n}} \in \mathscr{B}^{\epsilon} \left(A^n B^n\right)_\rho$. Then there exists a normalized Bell-diagonal quantum state $\tau_{A^{n}B^{n}} \in S_{=} \left(\mathcal{H}_{A^nB^n}  \right)$ such that $\tau_{A^{n}B^{n}} \in \mathscr{B}^{\epsilon} \left(A^n B^n\right)_\rho$ and $H_0 (A^nB^n)_{\tau} = H_0 (A^nB^n)_{\sigma}$.
\end{lem}
\begin{proof}
For any $\sigma_{A^{n}B^{n}} \in S_{\leq} \left(\mathcal{H}_{A^nB^n}  \right)$, let us define the normalized state $\tau_{A^{n}B^{n}} = \frac{\sigma_{A^{n}B^{n}}}{\operatorname{Tr}\left[\sigma_{A^{n}B^{n}}\right]}$. From this definition, it follows that both states have the same rank, i.e. $H_0 (A^nB^n)_{\tau} = H_0 (A^nB^n)_{\sigma}$. Moreover, 
\begin{eqnarray}
    \operatorname{F} (\rho,\tau) &=&\norm{\rho^{1/2}\tau^{ 1/2}}_{1}^{2} \\
    &=& \frac{1}{\operatorname{Tr}\left[\sigma\right]}\norm{\rho^{1/2}\sigma^{ 1/2}}_{1}^{2} \\
    &=& \frac{1}{\operatorname{Tr}\left[\sigma\right]}\operatorname{F} (\rho,\sigma) \\
    &\geq & \operatorname{F} (\rho,\sigma) \; .
\end{eqnarray}
From this, it follows that 
\begin{equation}
    \operatorname{P} (\rho,\tau) \leq \operatorname{P} (\rho,\sigma) \leq \epsilon \; .
\end{equation}
\end{proof}
Theorem \ref{Method: TheoremTightRate} gives a lower bound on the distillation rate, and it is derived by combining Proposition \ref{BenDistillationError} with Lemma \ref{NormalizedOptimizer}.
\begin{thm}  [Tight Bounds]\label{Method: TheoremTightRate}
Let $\rho_{A^nB^n} \in S_{=} \left(\mathcal{H}_{A^nB^n}  \right)$ be a normalized Bell-diagonal density matrix, whose eigenvalues are described by a probability distribution $P_{X^n}$. For all $\epsilon_1,\epsilon_2 >0$ that satisfy $\epsilon_{1}+\epsilon_{2} \leq \epsilon$,
\begin{equation} \label{Method: Eqtightbound}
    m^\epsilon \geq  n-\lceil H_{0}^{\epsilon_{1}}(X^{n})_{P}-2\log_{2}\left(\epsilon_{2}\right) \rceil \; ,
\end{equation}
and the rate of the one-way hashing method can be lower bounded by
\begin{equation} \label{Method: Eqtightrate}
    R^\epsilon \geq n-\frac{\ceil*{H_{0}^{\epsilon_{1}}(X^{n})_{P}-2\log_{2}\left(\epsilon_{2}\right)} }{n} \; .
\end{equation}
\end{thm}
\begin{proof}
    For any state $\rho_{A^nB^n} \in S_{=} \left(\mathcal{H}_{A^nB^n}  \right)$, let $\sigma_{A^{n}B^{n}}$ be a normalized state such that  $\sigma_{A^{n}B^{n}} \in \mathscr{B}^{\epsilon_1} \left(A^n B^n\right)_\rho$ and $H_0 (A^nB^n)_{\sigma} = H_{0}^{\epsilon_1} (X^n)_{P}$. Note that this state must exist due to the definitions for smooth entropies and Lemma \ref{NormalizedOptimizer}. After $r:=\lceil H_{0}^{\epsilon_{1}}(X^{n})_{P}-2\log_{2}\left(\epsilon_{2}\right) \rceil $ rounds, Proposition \ref{BenDistillationError} implies that, for $\sigma_{A^{n}B^{n}}$,
    \begin{align}
        \operatorname{F} (\sigma^{\prime}_{A^{n-r}B^{n-r}},\ketbra{\Phi_{A^{n-r} B^{n-r}}}{\Phi_{A^{n-r} B^{n-r}}} ) &\geq 1 - 2^{(H_0 (A^nB^n)_{\sigma} -r)} \\
        &= 1 - 2^{(H_{0}^{\epsilon_1} (X^n)_{P} -r)} \\
        &\geq 1 - 2^{2\log_{2}\left(\epsilon_{2}\right)} \\
        &= 1 - \epsilon_{2}^2
    \end{align}
    In terms of the purified distance, one then has that
    \begin{eqnarray}
              \operatorname{P} (\sigma^{\prime}_{A^{n-r}B^{n-r}},\ketbra{\Phi_{A^{n-r} B^{n-r}}}{\Phi_{A^{n-r} B^{n-r}}}) \leq \epsilon_{2}
    \end{eqnarray}
 Using the property that the purified distance satisfies the triangle inequality and is monotone under trace non-increasing completely positive maps, see e.g.~\cite[Proposition 3.1]{Tomamichel2015QuantumIP}, one has that
 \begin{align} \label{Eq: TriangleIneq}
     \operatorname{P} (\rho^{\prime}_{A^{n-r}B^{n-r}},\ketbra{\Phi}{\Phi} ) &\leq  \operatorname{P} (\rho_{A^nB^n},\sigma_{A^nB^n} ) + \operatorname{P} (\sigma^{\prime}_{A^{n-r}B^{n-r}},\ketbra{\Phi}{\Phi} ) \\
     &\leq \epsilon_1 + \epsilon_2 \\
     &\leq \epsilon \; ,
 \end{align}
 where $\ket{\Phi}$ again represents $\ket{\Phi_{A^{n-r} B^{n-r}}}$.
    Moreover, recall that Alice and Bob are left with $n-r$ Bell pairs and it thus holds that
    \begin{equation}
    m^\epsilon \geq  n-\ceil*{H_{0}^{\epsilon_{1}}(X^{n})_{P}-2\log_{2}\left(\epsilon_{2}\right)} \; ,
\end{equation}
as well as
\begin{equation}
    R^\epsilon \geq \frac{n-\ceil*{H_{0}^{\epsilon_{1}}(X^{n})_{P}-2\log_{2}\left(\epsilon_{2}\right) }}{n} \; .
\end{equation}
\end{proof}

%To calculate explicit bounds, one needs a method to compute the smooth Hartley entropy. Proposition \ref{CalcRank} gives an optimization problem that can be used to calculate the smooth Hartley entropy for Bell-diagonal states. Alternatively, on could relate the Hartley entropy to the max entropy and use the SDP in~\cite{tomamichel2013framework} to compute the the smooth max entropy. 
\begin{lem} \label{CalcRank}
Let $\mathcal{X}$ by a finite set and let $P_{X}(x)$ be a normalized probability distribution on $\mathcal{X}$. Then $H_{0}^{\epsilon} \left(X\right)_P$ is given by
\begin{equation} \label{optprobHart}
 \begin{aligned}
\min_{m} \quad & \log_{2} \left(m\right) \\
\textrm{s.t.} \quad & \sum_{x \in \mathcal{I}_{m}} P_{X}\left(x\right) \geq 1 - \epsilon^2 \; ,
\end{aligned}
\end{equation}
where $P_{X}\left(x_{1}\right), \dots,  P_{X}\left(x_{m}\right)$ are the $m$ largest weights of the distribution and $\mathcal{I}_{m} =\{ x_{1}, \dots,  x_{m}\}$
\end{lem}
\begin{proof}
    Let $m^\star$ be the optimal $m$ that satisfies Eq.~(\ref{optprobHart}). We will first show that $H_{0}^{\epsilon} \left(X\right)_P \leq \log_{2} \left(m^\star\right)$.
Let us consider the normalized distribution
\begin{align}
     Q_{X}(x) = \begin{cases} 
          \frac{P_{X}\left(x\right)}{\sum_{x \in \mathcal{I}_{m^\star} }P_{X}\left(x\right)} & \text{if} \ x \in \mathcal{I}_{m^\star} \\
          0 & \text{if} \ x \notin \mathcal{I}_{m^\star} \; .
       \end{cases} 
\end{align}
For this distribution, it holds that $H_{0} \left(X\right)_Q =\log_{2} \left(m^\star\right)$ and 

\begin{align}
    \operatorname{P} (P_X,Q_X )  &= \sqrt{1-\operatorname{F}  (P_X,Q_X)}  \\
    &= \sqrt{1-\left(\sum_{x \in \mathcal{I}_{m}} \sqrt{P_X(x) Q_X(x)} \right)^2} \\
    &= \sqrt{1-\sum_{x \in \mathcal{I}_{m^\star}} P_{X}\left(x\right)} \\
    &\leq  \epsilon  \; .
\end{align}
Since $Q_X$ is at least $\epsilon$-close to $P_X$ and has $H_{0} \left(X\right)_Q = \log_{2} \left(m^\star\right)$, it follows that 
$H_{0}^{\epsilon} \left(X\right)_P \leq \log_{2} \left(m^\star\right)$.
Let us now show the reverse inequality. Assume $H_{0}^{\epsilon} \left(X\right)_P = \log_{2} \left(m^\prime\right)$ where $m^\prime < m^\star$. The maximal achievable fidelity for any sub-normalized distribution with rank less than $ m^\star$ is given by
\begin{eqnarray}
\max_{ Q_{X}(x) \in S_{\leq} \left(X\right), \ \text{rank}[Q]<m^\star} \left(\sum \sqrt{P_{X}\left(x\right)Q_{X}\left(x\right)} \right)^2 \; .
\end{eqnarray}
W.l.o.g. assume that $P_X(1) \geq P_X(2) \geq \dots \geq P_X(|\mathcal{X}|)$. Then there exists a \textit{normalized} distribution $Q_X$ that maximizes this expression for which the only potential non-trivial entries are given by $Q_X(1),\dots, Q_X(m^\star-1)$. In this case, the maximization is equivalent to
\begin{eqnarray}
\max_{  \sum_{i=1}^{m^\star -1} Q_{X}\left(i\right) = 1} \left(\sum_{x=1}^{m^\star -1} \sqrt{P_{X}\left(x\right)Q_{X}\left(x\right)}\right)^2 \; .
\end{eqnarray}
This optimization problem can be solved via Lagrange multipliers, i.e. one has to optimize the function
\begin{eqnarray}
\sum_{x=1}^{m^\star -1} \sqrt{P_{X}\left(x\right)Q_{X}\left(x\right)}- \lambda \left(\sum_{x = 1}^{m^\star -1} Q_{x}\left(x\right) -1\right) \; .
\end{eqnarray}
The conditions that the optimal solution has to satisfy are given by
\begin{eqnarray}
\frac{\sqrt{P_{X}\left(x\right)}}{2\sqrt{Q_{X}\left(x\right)}} &=& \lambda \quad \forall x \in {1,\dots, m^\star -1}\\
\sum_{x = 1}^{m^\star -1} Q_{X}\left(x\right) &=& 1 \; .
\end{eqnarray}
Combining these conditions yields the constraint
\begin{eqnarray}
\sum_{x=1}^{m^\star -1}  \frac{P_{X}\left(x\right)}{4\lambda^2} = 1 \; .
\end{eqnarray}
We thus have that 
\begin{eqnarray}
    4\lambda^2 = \sum_{x=1}^{m^\star-1} P_{Y}\left(i\right) \; .
\end{eqnarray}
The optimal solution for $Q_X$ is therefore given by
\begin{eqnarray}
Q_{X}\left(x\right) = \frac{P_{X}\left(x\right)}{\sum_{x=1}^{m^\star-1} P_{X}\left(x\right)} \quad \forall x \in {1,\dots, m^\star-1} \; ,
\end{eqnarray}
and the maximal achievable fidelity is $\sum_{x=1}^{m^\star-1} P_{X}\left(x\right)$. However, since $m^\star$ is the optimal solution to Eq.~(\ref{optprobHart}), it must follow that 
\begin{eqnarray}
    \operatorname{F}  (P_X,Q_X) = \sum_{x=1}^{m^\star -1} P_{X}\left(x\right) < 1-\epsilon^2 \; .
\end{eqnarray}
This is equivalent to the inequality
\begin{eqnarray}
    \operatorname{P}  (P_X,Q_X) > \epsilon \; , 
\end{eqnarray}
which implies that there exists no distribution with rank $m^\prime < m^\star$ that is $\epsilon$-close to $P_X$, and thus ${H_{0}^{\epsilon} \left(X\right)_P \geq \log_{2} \left(m^\star\right)}$.
\end{proof}

We now prove the asymptotic behavior of our lower bound for the rate, using the classical AEP from~\cite{holenstein2006randomness}. In~\cite{holenstein2006randomness}, they consider a variation of the smoothed Hartley entropy, which implicitly uses the generalized trace distance (for a formal definition see~\cite[Definition 3.4]{Tomamichel2015QuantumIP}) as a distance measure rather than the purified distance. These distance measures can be related to each other via Fuchs-van de Graaf-like inequalities~\cite[Lemma 3.5]{Tomamichel2015QuantumIP}. In particular one then has that
\begin{align} \label{Eq: SanwichDifferentHartleys}
   \tilde{H}_{0}^{\epsilon}(X^{n})_{P} \leq H_{0}^{\epsilon}(X^{n})_{P} \leq \tilde{H}_{0}^{\epsilon^2/2}(X^{n})_{P} \; ,
\end{align}
where $\tilde{H}_{0}^{\epsilon}(X^{n})_{P}$ denotes an alternate definition for the smoothed Hartley entropy from~\cite{holenstein2006randomness}, which uses the generalized trace distance.
\begin{lem}  
Let $ \rho_{AB}^{\otimes n} \in S_{=} \left(\mathcal{H}_{A^nB^n}  \right)$ be a normalized Bell-diagonal density matrix, whose eigenvalues are described by a probability distribution $P_{X^n}$. Then
\begin{equation}
    \lim_{\epsilon_1, \epsilon_2 \to 0}  \lim_{n \to \infty}  \frac{n-\ceil*{H_{0}^{\epsilon_{1}}(X^{n})_{P}-2\log_{2}\left(\epsilon_{2}\right) }}{n} = -H(A|B)_\rho \; .
\end{equation}
\end{lem}
\begin{proof}
    As was shown in~\cite[Lemma 3]{RenWol05}, the results from~\cite{holenstein2006randomness} imply that 
\begin{align}
   \lim_{\epsilon \to 0}  \lim_{n \to \infty} \frac{\tilde{H}_{0}^{\epsilon}(X^{n})_{P}}{n} = \lim_{\epsilon \to 0}  \lim_{n \to \infty} \frac{\tilde{H}_{0}^{\epsilon^2/2}(X^{n})_{P}}{n} =  H(X)_{P}\; .
\end{align}
From Eq.~(\ref{Eq: SanwichDifferentHartleys}) and $H(X)_{P} = H(AB)_\rho$, it then follows that 
\begin{align}
   \lim_{\epsilon_1, \epsilon_2 \to 0}  \lim_{n \to \infty}  \frac{n-\ceil*{H_{0}^{\epsilon_{1}}(X^{n})_{P}-2\log_{2}\left(\epsilon_{2}\right) }}{n} &= 1 - H(X)_{P} \\
   & = 1 - H(AB)_\rho \\
   & = -H(A|B)_\rho \; .
\end{align}
Moreover, the last equation holds due to the chain rule ${H(A|B)_\rho = H(AB)_\rho - H(B)_\rho}$ and the fact that $H(B)_{\rho} =1$ for Bell-diagonal states. 
\end{proof}
\subsection{Modified One-Way Hashing Method} \label{Subsection: ModifiedProtocol}
\begin{tcolorbox}[size=small, title=Modified One-Way Hashing Method]
Single Round (on $n$ qubits): 
\begin{steps} [leftmargin=1.4cm] %\label{EntDist}
  \item Alice and Bob generate a uniformly random bitstring $S = s_{1}\cdots s_{n-k}$ where each $s_{j}$ represents two bits. Let $s_{j^{\star}}$ represent the first non-zero 2-bit string.%\footnote{If no such 2-bit string exists, one can either generate a new bitstring $S$ or randomly select a value for $j^{\star}$, and jump directly to Step 4.}
  \item For all $j \in \{ 1,\dots, n-k \}$, if
  \SubItem {$s_{j} = 00$, no actions are required at Step 2.}
   \SubItem {$s_{j} = 01$ and $j \neq j^\star$, no actions are required at Step 2. If $j =  j^\star$, then both Alice and Bob apply a $\pi/2$-rotation around the y-axis on their half of the $j$'th qubit pair.}
  \SubItem {$s_{j} = 10$ and $j \neq j^\star$ then both Alice and Bob apply a $\pi/2$-rotation around the y-axis on their half of the $j$'th qubit pair. If $j =  j^\star$, no actions are required at Step 2.}
  \SubItem {$s_{j} = 11$ and $j \neq j^\star$, then Alice and Bob, respectively, apply a $3\pi/2$- and $\pi/2$-rotation around the x-axis on their half of the $j$'th qubit pair. If $j = j^\star$, then Alice and Bob, respectively, apply a $3\pi/2$- and $\pi/2$-rotation around the z-axis on their half of the $j$'th qubit pair. }
  \item For all $j \neq j^\star$ s.t. $s_{j} \neq 00$, both Alice and Bob apply a CZ gate on qubit pairs $j$ and $j^\star$, where pair $j^\star$ contains the target qubits.
  \item Alice and Bob both apply a $\pi/2$-rotation around the y-axis on their half of the target qubit pair $j^\star$. Then they measure the target qubit pair in the computational basis and discard said pair. Alice broadcasts her measurement outcomes to Bob.
\end{steps}
Post-processing: After all rounds are concluded, Alice and Bob share a Bell-diagonal state. Based on the initial (pre-protocol) bipartite state and all past joint measurement outcomes, Bob applies single-qubit Pauli gates such that the largest eigenvalue of the post-processed state corresponds to multiple copies of the Bell state $\ketbra{\phi^{+}}{\phi^{+}}_{AB}$. 

\end{tcolorbox}

Acting on $n$ qubit pairs, a single round of the original one-way hashing method first generates a uniformly random string ${S = s_1 \cdots s_n}$. For $n$ qubit pairs, any (single) error we consider can be described as a $2n$-bitstring, ${X = x_1 \cdots x_n}$, using the following convention. If $x_i = ab$, then the error on the $i$'th pair is described by the Pauli gate $Z_{A}^a X_{A}^b$, which acts solely on Alice's system, and the resulting $i$'th pair is given by the Bell state $  Z^a_{A} X^b_{A} \otimes \mathds{1}_{A}\ket{\phi^{+}}$. Due to this connection, one typically labels the four potential Bell states via: 
\begin{align}
\begin{aligned}
    00 \mapsto \ket{\phi^{+}}:= \frac{1}{2}  \left[ \ket{00} + \ket{11}\right]\\
    10 \mapsto \ket{\phi^{-}}:= \frac{1}{2}  \left[ \ket{00} - \ket{11}\right] \\
    01 \mapsto \ket{\psi^{+}}:= \frac{1}{2}  \left[ \ket{01} + \ket{10}\right]\\
    11 \mapsto \ket{\psi^{-}}:= \frac{1}{2}  \left[ \ket{01} - \ket{10}\right] 
\end{aligned}
\end{align}
The first bit is referred to as the ``phase" bit and the later is the ``amplitude" bit. The key insight from~\cite{Bennett1996MixedstateEA} is that if a single round of the one-way hashing method acts on a state with error $X$, then the Boolean inner product
\begin{align}
    S \cdot X
\end{align}
can be determined from the measurement outcomes of that single round. As such, for any observed value $S \cdot X$, Alice and Bob will discard any error which could not have produced this output. It was shown in~\cite{Bennett1996MixedstateEA} that, on average, half of the errors will be discarded per round, and we implicitly use this fact in Proposition \ref{BenDistillationError}. 

The modified one-way hashing method can also be used to calculate the Boolean inner product and therefore the analytical bounds from this work hold for this protocol as well. To see how $S \cdot X$ is determined, we now discuss how the modified protocol acts on any error $X$. For $j\neq j^\star$, the protocol remains identical to the one-way hashing in the first two steps. As such, the analysis from~\cite{Bennett1996MixedstateEA} still holds, and for any $j \neq j^\star$ and $s_{j} \neq 00$, the information $s_{j} \cdot x_{j}$ is stored in the amplitude bit of the $j$-th qubit pair. For the target pair, it can readily be verified that after Step $2$ the value $s_{j^\star} \cdot x_{j^\star}$ is stored in the phase bit of the target qubit pair. 

Alice and Bob applying CZ gates between the pairs $j$ and $j^\star$ is described by the mapping
\begin{align}
\begin{aligned}
    \text{Target:}& \ a_{j^\star}b_{j^\star} \mapsto \left(a_{j^\star} \oplus b_{j}\right) b_{j^\star}\\
\text{Control:}& \ a_{j} \ b_{j} \: \,   \mapsto \left(a_{j} \oplus b_{j^\star}\right) b_j \; .
\end{aligned}
\end{align}
In particular, this ensures that, after Step $3$, the relevant amplitude bits are added to the phase bit of the target pair. By virtue of the process done in Step $2$, it then follows that the phase bit of the target pair is equal to $S \cdot X$. The bilateral rotation in Step $4$ simply ensures that the phase and amplitude bit of the target pair are flipped. If $S \cdot X = 0$, both parties achieve the same measurement output. Conversely, if $S \cdot X = 1$, then Alice and Bob measure opposite outcomes.
\subsection{Connection Between One-Way Hashing Method and Quantum Error Correcting Codes} \label{ECCConnection}
For $n=5$, one can for example correct all first order errors by using the strings
\begin{eqnarray}
    S_1 &=& 01 \ 01 \ 01 \ 01 \ 00 \\
    S_2 &=& 10 \ 10 \ 10 \ 00 \\
    S_3 &=& 01 \ 11 \ 01 \\
    S_4 &=& 01 \ 10 \; .
\end{eqnarray}
Note that, akin to the $[[5,1,3]]$ QECC, this is the lowest number of pairs for which all first-order errors can be corrected. Conversely, if one only wants to detect all first-order errors, the corresponding strings are 
\begin{eqnarray}
    S_1 &=& 11 \ 11 \ 11 \ 11 \\
    S_2 &=& 11 \ 11 \ 11  \; .
\end{eqnarray}
If no errors are detected, both parties will share two Bell pairs with a global fidelity of order $1-O\left(\left(1-W\right)^2\right)$. Analogously to the recurrence method, they discard the post-measurement state if an error is detected. At least $4$ pairs are needed to detect all first-order errors, as it works analogously to the $[[4,2,2]]$ stabilizer code.
 \hfill

 \hfill

  \hfill

 \hfill

  \hfill

 \hfill

  \hfill

 \hfill

  \hfill

 \hfill

 \hfill
%\pagebreak
%\printbibliography
\bibliography{apssamp}% Produces the bibliography via BibTeX.

\end{document}